\documentclass[11pt]{article}

\usepackage{a4}                        	
\usepackage{amssymb}                   	
\usepackage{amsmath}                   	
\usepackage{amsthm}                    	
\usepackage{cite}

\newtheorem{theorem}{Theorem}
\newtheorem*{theorem*}{Theorem}
\newtheorem{lemma}{Lemma}

\newtheorem{corollary}{Corollary}

\theoremstyle{definition}

\theoremstyle{remark}

\DeclareMathOperator{\tr}{tr}
\DeclareMathOperator{\diag}{diag}

\begin{document}

\begin{center}
\LARGE{The finite Fourier Transform and projective 2-designs} 
\end{center}
\vspace{0.3 cm}

\begin{center}
\large{Gerhard Zauner} \\ 
\vspace{1mm} 
\footnotesize{mail@gerhardzauner.at} 
\end{center}
\bigskip
\bigskip
\bigskip

\begin{abstract}
There are several approaches to define an eigenvector decomposition of the \emph{finite Fourier Transform} (Fourier matrix), which is in some sense unique, and at best resembles the eigenstates of the quantum harmonic oscillator.

A solution given by Balian and Itzykson~\cite{BI86} in 1986 for prime dimensions $d=3$ (mod $4$) is revisited. 
It is shown, that by applying the Weyl-Heisenberg matrices to this eigenvector basis, a \emph{projective $2$-design} is generated.
\end{abstract}
\bigskip

\section{Motivation: The Quantum Harmonic Oscillator}

The position operator $\mathbf{X}$ and the impulse operator $\mathbf{P}$ 
are each defined on a dense subset of
$\mathcal{L}^2(\mathbb{R})$ via the equations
$(\mathbf{X}f)(x) =
xf(x)$, and $\mathbf{P}f)(x) = -i\frac{d}{dx}f(x)$.

Here we have set $\hbar=1$.
They fulfill the \emph{Canonical commutation relation} $ \left[\mathbf{X},\mathbf{P}\right]= i \mathbf{I}$. 
When we set all physical parameters to $1$, the \emph{Hamiltonian} for the quantum harmonic oscillator is
$\mathbf{H} = \frac{1}{2}(\mathbf{X}^2+\mathbf{P}^2)$.
Its \emph{eigenvalues} resp. energy levels are $n+ \frac{1}{2}$ with $n=0,1,2, \ldots $. The corresponding \emph{eigenstates} are
$
\psi_n(x) = \frac{1}{\sqrt{2^n n! \sqrt{\pi}}} H_n(x)e^{\frac{-x^2}{2}} 
$
with the \emph{Hermite polynomials} $H_n(x)=(-1)^n e^{x^2}\frac{d^n}{dx^n}(e^{-x^2})$. All $\psi_n$ constitute also a particular choice of eigenstates for the \emph{Fourier-Transform} 
\begin{equation*}
(\mathbf{F}f)(y) = \frac{1}{\sqrt{2 \pi}}\int_{-\infty}^{\infty}e^{-i 2 \pi x y}f(x) dx.  
\end{equation*}
$\mathbf{F}$ has the $4$ eigenvalues $\pm 1, \pm i$:
$(\mathbf{F}\psi_n)(x) = (-i)^n \psi_n(x)$.
\medskip

To $\mathbf{X}$ and $\mathbf{P}$ a two-parameter strongly continuous group can be assigned, the so called \emph{Weyl-Heisenberg} group of unitary operators
\begin{equation*}                             
\boldsymbol{\mathsf{W}}(r,s) := e^{i(r\mathbf{P}+s\mathbf{X})}=
e^{-\frac{i r s}{2}}e^{ir\mathbf{P}} e^{is\mathbf{X}}.
\end{equation*}
With 
$r,s \in \mathbb{R} \to \boldsymbol{\mathsf{W}}(r,s)$ they make up the unique unitary, irreducible, projective representation of the additive group $\mathbb{R} \times \mathbb{R}$,~\cite{Th79}. 

\bigskip
Now we turn to the finite dimensional counterparts.

\section{Weyl-Heisenberg Matrices and Fourier Matrix}

Let $d$ be the dimension of a finite-dimensional complex vector space.
\bigskip

We use throughout the paper the notation of bra's $\langle\mathbf{.}|$ (row-vectors) and ket's $|\mathbf{.}\rangle$ (column-vectors).
Let $|\mathbf{e}_r\rangle$, with $r \in \mathbb{Z}_d$ be the standard basis.
Addition of indices is always \emph{modulo} $d$, so e.g. for bra's $|\mathbf{e}_{r+s}\rangle:=|\mathbf{e}_{(r+s) (mod \: d)}\rangle$.  
\bigskip

Let 
$\mathbf{U}|\mathbf{e}_r\rangle = e^{\frac{i2\pi r}{d}}|\mathbf{e}_r\rangle$ and
$\mathbf{V}|\mathbf{e}_r\rangle = |\mathbf{e}_{r+1}\rangle$ or in matrix form:
\begin{equation*}
\mathbf{U} =
\begin{pmatrix}
 1     &  0                           &  0                            & \ldots &  0     \\
 0     &  e^{\frac{i2\pi}{d}}  &  0                             & \ldots &  0     \\
 0     &  0                           &   e^{\frac{i4\pi}{d}}   & \ldots &  0     \\
\vdots & \vdots                  & \vdots                      & \ddots & \vdots \\
 0     &  0                           &  0                             & \ldots &  e^{\frac{i2(d-1)\pi}{d}}
\end{pmatrix},
\quad
\mathbf{V} =
\begin{pmatrix}
 0     &  0     &  0     & \ldots  &  0  &  1     \\
 1     &  0     &  0     & \ldots  &  0  &  0     \\
 0     &  1     &  0     & \ldots  &  0  &  0     \\
 0     &  0     &  1     & \ldots  &  0  &  0     \\
\vdots & \vdots & \vdots  & \ddots  & \vdots & \vdots \\
 0     &  0     &  0     & \ldots  &  1  &  0  
\end{pmatrix}.
\end{equation*}
Let $\tau=e^{\frac{i \pi (d+1)}{d}}=-e^{\frac{i\pi}{d}}$, and for $r, s \in \mathbb{Z}_d$ for odd d, resp.  
$r, s \in \mathbb{Z}_{2d}$  for even d
\begin{equation*}
\boldsymbol{\mathsf{W}}_{(r,s)} := \tau^{r s} {\mathbf{V}}^{r}{\mathbf{U}}^{s}
\end{equation*}

The case of even $d$ needs $2d \times 2d$ matrices due to the phase added. We will look only at the odd case in the following.

These matrices generate the  so-called finite \emph{Weyl-Heisenberg} group, and make up the unique, irreducible, projective representation of the additive group $\mathbb{Z}_d \times \mathbb{Z}_d$.

The factor $\tau^{r s}$ is the analogue of the factor $e^{-\frac{i r s}{2}}$ in the infinite dimensional case and simplifies calculations significantly ({\sc Appleby}~\cite{Ap05}). 
\medskip

In some papers $\dot{\tau}=e^{\frac{i \pi (d^2+1)}{d}}=(-1)^d e^{\frac{i\pi}{d}}$ is used instead of $\tau$. Sometimes inverse matrices are used in the definition. Furthermore for odd $d$ the matrix
$\mathbf{U'} =\diag (e^{\frac{-i \pi (d-1)}{d}},\ldots,  e^{\frac{-i2\pi}{d}}, 1, e^{\frac{i2\pi}{d}}, \ldots, e^{\frac{i \pi (d-1)}{d}})$ is often used, as it reflects the axial symmetry in the infinite case (e.g. {\sc Singh} and {\sc Carroll}~\cite{SC18}).{\footnote{
\emph{\textbf{Remark  on the notation:}} 
Due to the fact that  $\mathbf{U}$ and $\mathbf{V}$ can also be considered as generalizations of the \emph{Pauli Matrices} $\mathbf{Z}, \mathbf{X}$, defined in dimension $d=2$ (as  rotations of the Bloch Sphere around the Z- and X-axis), 
in several papers these letters are used instead.
We use the letters $\mathbf{U}$ and $\mathbf{V}$ as e.g. {\sc Schwinger}~\cite{Sc60} and in many subsequent physics papers, to emphasize the fact that we work on qudits ($d \geq 2$) and not just qubits. 

The letter $\boldsymbol{\mathsf{W}}$ stands for {\sc Weyl}, who brought the matrices
above onstage in physics first~\cite{We31}.
Very often $\mathbf{D}$ for \emph{Displacement} is used instead of $\boldsymbol{\mathsf{W}}$.
}}

\bigskip

In contrast to the infinite case, we don't have uniquely determined infinitesimal generators of the Weyl-Heisenberg group in finite dimensions (as $\mathbf{X}, \mathbf{P}$). 
Therefore we have also no uniquely defined finite counterpart of the quantum harmonic oscillator and its eigenstates. 
\newpage

But we have at least a finite counterpart to the Fourier-Transform:
The \emph{Fourier-Matrix} (also called \emph{Finite} or  \emph{Discrete
Fourier-Transformation} or \emph{Schur-Matrix}) is the $d \times d$ matrix
\begin{equation*}
\mathbf{F}=
\frac{1}{\sqrt{d}} \sum_{r=0}^{d-1} \sum_{s=0}^{d-1}e^{\frac{2 i  \pi r s}{d}}|\mathbf{e}_{r}\rangle\langle\mathbf{e}_{s}|.
\end{equation*}
In matrix form
\begin{equation*}                                   
\mathbf{F} = {\frac{1}{\sqrt{d}}}
\begin{pmatrix}
1      & 1                    			& 1                     			& \ldots & 1                          			\\
1      & e^{\frac{2{i \pi}}{d}}        	& e^{\frac{4{i \pi}}{d}}         	& \ldots & e^{\frac{2(d-1){i \pi}}{d}}     		\\
1      & e^{\frac{4{i \pi}}{d}}       	&e^{\frac{8{i \pi}}{d}}       		& \ldots & e^{\frac{4(d-1){i \pi}}{d}}       	\\
\vdots & \vdots               		& \vdots                			& \ddots & \vdots                     			\\
1      & e^{\frac{2(d-1){i \pi}}{d}} 	& e^{\frac{4(d-1){i \pi}}{d}}  	& \ldots & e^{\frac{(d-1)^2{i \pi}}{d}} 
\end{pmatrix}.
\end{equation*}

As $\mathbf{F}^4 = \mathbf{I}$, the possible eigenvalues of $\mathbf{F}$ are $\pm 1, \pm i$, as in the infinite case. 
The multiplicity of the eigenvalues as function of the dimension $d$ is given by the following table (see~\cite{AT79}).
\begin{equation*}                                        \label{E:Z-EW}
   \begin{tabular}{|c|c|c|c|c|c|} \hline
           		&  $1$    	&  $-1$	&  $i$      	&  $-i$               \\ \hline \hline
      $d=4k$    	&  $k+1$  	&  $k$      	&  $k$      	&  $k-1$              \\ \hline
      $d=4k+1$  	&  $k+1$  	&  $k$      	&  $k$      	&  $k$                \\ \hline
      $d=4k+2$  	&  $k+1$  	&  $k+1$  	&  $k$      	&  $k$                \\ \hline
      $d=4k+3$  	&  $k+1$  	&  $k+1$  	&  $k+1$ 	&  $k$                \\ \hline
   \end{tabular}
\end{equation*}
\bigskip

There is long and ongoing history to define an, in some sense unique, eigenvector decomposition of the Fourier Matrix.

Proposals come from Mathematicians, Physicists, and Electrical and Electronics Engineers (see examples in~\cite{ASSW14, ABD14, BI86, DS82, FL97, Gr82, GH08, KK17, MP72, Mo80, WS07, ZV05}).

The approaches are either starting with the continuous harmonic oscillator eigenstates (e.g. sampling at equidistant points), or focusing on algebraic methods.
\medskip

Here we propose a characterization in terms of Quantum Designs, specifically projective $2$\emph{-designs}.

\section{Projective $2$-designs}

Let $\{|\psi_i\rangle: 1 \leq i \leq N\}$ be $N$ normed vectors and $\{\mathbf{P}_i:=|\psi_i\rangle \langle \psi_i| : 1 \leq i \leq N\}$ be the  
corresponding projection matrices. 
There are many equivalent definitions, when these sets form a projective $2$-\emph{design}, see e.g.~\cite{BZ17a,BZ17b, DGS77, Ho82, RBSC03, RS07,Wa18, Za99}, 
(historically) starting with the condition, that
\begin{equation*}                            
\frac{1}{N}\sum_{i=1}^N f(\mathbf{P}_i)  = \int_{\mathbb{C}\mathbf{P}^{n-1}}f(\mathbf{P}) dp.
\end{equation*}
for any homogeneous polynomial $f$ of degree $2$, and the normed unitary invariant (Haar-)integral on the right side. 

This is equivalent to
\begin{equation}       \label{E:2-Des-Prop}                       
\frac{1}{N}\sum_{i=1}^N \mathbf{P}_i \otimes \mathbf{P}_i = \frac{2}{d(d+1)} \mathbf{\Pi}_{\text{sym}} 
\end{equation}
where $\mathbf{\Pi}_{\text{sym}}$ is the orthogonal projection on the symmetric subspace of $\mathbb{C}^d \otimes \mathbb{C}^d$.
Another equivalent condition is, that the inequality
\begin{equation*}                             
\frac{1}{N^2}\sum_{i=1}^N\sum_{j=1}^N (tr(\mathbf{P}_i \mathbf{P}_j)^2 = \frac{1}{N^2}\sum_{i=1}^N\sum_{j=1}^N | \langle \psi_i | \psi_j \rangle |^4 \geq \frac{2}{d(d+1)}
\end{equation*}
becomes an equality. Well know examples are ~\cite{BZ17a, BZ17b}:
\begin{itemize}
\item{SIC} - (conjectured to exist for all $d \in \mathbb{N}$): $N=d^2$ ,
\item{Complete sets of MUBs} - (exist for all $d=$ prime power): $N=d(d+1)$ 
\item{Clifford group applied to any vector (for $d=$ prime}): $N=d(d^2-1)$  
\end{itemize}
\smallskip
In the last example the Clifford group for prime $d$ is also an example of a so-called \emph{unitary} $2$-design. We only refer to \emph{projective $2$-designs} in this paper and skip the attribute \emph{projective} frequently. 

\section{An example for $d=3$ and some numerical search}

The Fourier matrix for d=3
\begin{equation*}                                   
\mathbf{F} = {\frac{1}{\sqrt{3}}}
\begin{pmatrix}
1      & 1              & 1           \\
1      & \alpha        & \alpha^2   \\
1      & \alpha^2      & \alpha
\end{pmatrix}, \qquad
\alpha = e^{2{\pi}i/3}.
\end{equation*}

has the $3$ eigenvalues $\pm 1$ and $i$, and the (up to phases) \emph{unique} normed eigenvectors are
\begin{equation*}
| \psi_{1}  \rangle = \frac{1}{\sqrt{6 + 2\sqrt{3}}}
\begin{pmatrix}
1  + \sqrt{3}\\
1  \\
1
\end{pmatrix},
|\psi_{-1}  \rangle = \frac{1}{\sqrt{6 - 2\sqrt{3}}}
\begin{pmatrix}
1 - \sqrt{3}\\
1  \\
1
\end{pmatrix},
|\psi_{i}  \rangle = \frac{1}{\sqrt{2}}
\begin{pmatrix}
0 \\
1 \\
-1
\end{pmatrix}
\end{equation*}
We apply the $d^2$ Weyl-Heisenberg matrices to each of the $d=3$ vectors and get a set of $d^3=27$ vectors
$\{\boldsymbol{\mathsf{W}}_{(r,s)} |\psi_x \rangle: x=\pm 1, i \mbox{ and } 0 \leq r,s, \leq 2 \}$.

These vectors form a $2$-design! Actually $|\psi_i \rangle$ is also a fiducial vector for a \emph{SIC-POVM}, 
so due to the additivity of the $2$-design property the first $2$ eigenvectors generate also a $2$-design of $18$ vectors.
\medskip

Numerical search found no complete eigenvector basis of the Fourier matrix for each $d=4,5,6$, that generates a $2$-design like above.
But for $d=7,11$ there were found in each case  (seemingly unique) solutions. 

For $d=5$ there is a (unique) basis of $2$ eigenvectors of eigenvalue $1$, which together with the unique eigenvector for eigenvalue $-1$ generates a set of
$3d^2=75$ vectors, that form a 2-design. 
\medskip

In all these case the eigenvectors can be taken to be real-valued only.

\newpage
To construct the appropriate eigenvector basis for primes $d=4k+3=3,7,11, ...$ we need some preparation.

\section{Clifford Group}
In the following section $d$ is taken to be \emph{odd}.{\footnote{For the definition of the restricted Clifford group given here the choice of the phases $\tau^{r s}$ for $\boldsymbol{\mathsf{W}}_{(r,s)}$ is essential. Therefore again in case of even $d$, one would have to choose $\mathbb{Z}_{2d}$ instead of $\mathbb{Z}_{d}$.}} 
\smallskip

The  \emph{Clifford group}  is defined as \emph{normalizer} 
of the Weyl-Heisenberg group $\tau^q \boldsymbol{\mathsf{W}}_{(r,s)}$, $0 \leq q,r,s \leq d-1$, in the group of unitary matrices $\mathbf{U} \in \mathcal{U}(d)$). It includes the Weyl-Heisenberg group itself as normal subgroup. We are interested here in the subgroup of $\mathbf{U} \in \mathcal{U}(d)$, for which for all $r,s \in \mathbb{Z}_d$
\begin{equation*}
\mathbf{U} \boldsymbol{\mathsf{W}}_{(r,s)} \mathbf{U}^{-1} 
= \boldsymbol{\mathsf{W}}_{(r',s')}
\end{equation*}
for some $r',s' \in \mathbb{Z}_d$ (with no additional phase involved!), and call it  
restricted Clifford group $\mathcal{C}(d)$.  
It is well known~\cite{Ap05, ABD14, BZ17a,BZ17b}, that  the restricted Clifford group is a projective representation of $SL(2,\mathbb{Z}_d)$
with the group homomorphism $h$  defined via 
\begin{equation*}
G=
\begin{bmatrix}
\alpha & \beta \\
\gamma & \delta
\end{bmatrix} 
\xrightarrow[h]{}
\mathbf{U}_G 
\quad \mbox {iff } \quad
\begin{bmatrix}
r'  \\
s' 
\end{bmatrix} =
\begin{bmatrix}
\alpha & \beta \\
\gamma & \delta
\end{bmatrix}
\begin{bmatrix}
r  \\
s 
\end{bmatrix} 
\end{equation*}
where $\mathbf{U}_G$ is a \emph{representative}, which is fixed up to an overall phase factor. 

 The Fourier matrix is an element of  the Clifford group:  
\begin{equation*}
{\mathbf{F}} \boldsymbol{\mathsf{W}}_{(r,s)} \mathbf{F}^{-1}
=\boldsymbol{\mathsf{W}}_{(-s,r)}
\quad \Rightarrow \quad
\begin{bmatrix}
0 & -1 \\
1 & 0
\end{bmatrix}  
\xrightarrow[h]{}
\mathbf{F}  
\end{equation*}

\subsection{The Subgroup of elements commuting with $\mathbf{F}$}

Let $\mathcal{FC}(d) \subset \mathcal{C}(d)$ be the subgroup of elements of the restricted Clifford group, that commute with the Fourier matrix. In terms of $SL(2,\mathbb{Z}_d)$ this means

\begin{equation*}
\begin{bmatrix}
\alpha & \beta \\
\gamma & \delta
\end{bmatrix}
\begin{bmatrix}
0 & -1 \\
1 & 0
\end{bmatrix} 
= 
\begin{bmatrix}
0 & -1 \\
1 & 0
\end{bmatrix} 
\begin{bmatrix}
\alpha & \beta \\
\gamma & \delta
\end{bmatrix} \quad
 \Longrightarrow \quad
\alpha= \delta, \gamma=-\beta. 
\end{equation*}
Therefore $\mathcal{FC}(d)= \{e^{i \xi}\mathbf{U}_G\}$ with $G \in FSL(2,\mathbb{Z}_d)$
\begin{equation*}
FSL(2,\mathbb{Z}_d)=
\left\{G=
\begin{bmatrix}
\alpha & \beta \\
-\beta & \alpha
\end{bmatrix}
:
\mbox{det}(G)= \alpha^2+\beta^2=1\right\} 
\end{equation*}
$\mathcal{FC}(d)$ rsp. $FSL(2,\mathbb{Z}_d)$ are \emph{abelian} (commutative) groups.
\medskip

If $d$ is an \emph{odd prime} then $FSL(2,\mathbb{Z}_d)$ is a cyclic group (see~\cite{BI86,AFN96, AFN97, FL97}) with order
\begin{align*}
|FSL(2,\mathbb{Z}_d)| =
\begin{cases}
d-1  & \mbox{if} \quad d=4k+1  \\ 
d+1 & \mbox{if} \quad d=4k+3 .
\end{cases}
\end{align*}

\newpage
\subsection{Representations of $\mathcal{C}(d)$ and especially $\mathcal{FC}(d)$ in $\mathbb{C}^d$}
We  assume $d$ to be an odd prime. An explicit representation of $\mathcal{C}(d)$ is given by 
{\sc Appleby}~\cite{Ap09} (see also~\cite{Ne02}). We restrict to the subgroup $FSL(2,\mathbb{Z}_d)$, and get, except when $\beta=0$ (which means in $FSL(2,\mathbb{Z}_d)$: $\alpha=\pm 1 \to G=\pm \mathbf{I}$).

\begin{equation*}
G=
\begin{bmatrix}
\alpha & \beta \\
-\beta & \alpha
\end{bmatrix}
\xrightarrow[h]{}
\mathbf{U}_G=
\frac{e^{i \theta}}{\sqrt{d}}\sum_{r=0}^{d-1}\sum_{s=0}^{d-1}\tau^{ \left(\frac{1}{\beta}\right) (\alpha r^2 - 2 r s + \alpha s^2)} |\mathbf{e}_{r}\rangle\langle\mathbf{e}_{s}|
\end{equation*}

The Weyl-Heisenberg matrices are an orthogonal basis of all matrices.
Therefore we can also expand the elements above in this basis. An according representation (see ~\cite{BI86}, and {\sc Athanasiu, Floratos, Nicolis}~\cite{AFN97}) is (except the case $\alpha=1, \beta=0$).
\begin{equation}  \label{E:R-input}
G=
\begin{bmatrix}
\alpha & \beta \\
-\beta & \alpha
\end{bmatrix}
\xrightarrow[h]{}
\mathbf{U}_G=
\frac{e^{i \eta}}{d}\sum_{r=0}^{d-1}\sum_{s=0}^{d-1}\tau^{\frac{\beta}{2(1-\alpha)} (r^2  + s^2)}\boldsymbol{\mathsf{W}}_{(r,s)} 
\end{equation}
The overall phases factors $e^{i \theta}$ rsp. $e^{i \eta}$ can be chosen such, that the presentation becomes de-projectivized (ordinary rsp. faithfull). We don't need this here, and set the factors to $1$.
\medskip

\section{2-designs from the Fourier matrix}

From here on we assume $d$ to be prime with $d= 4k+3$.
\medskip

In this case equation~\eqref{E:R-input}, with $e^{i \eta}=1$, simplifies to the $d$ matrices
\begin{equation} \label{E:Rm}
\mathbf{R}_m := \frac{1}{d} \sum_{r=0}^{d-1}\sum_{s=0}^{d-1}  \tau^{m (r^2  + s^2)} \boldsymbol{\mathsf{W}}_{(r,s)}, \quad  0 \leq m \leq d-1
\end{equation}

We get this form, when we set (mod $d$)
\begin{equation*}
m=\frac{\beta}{2(1-\alpha)} \iff \alpha=\frac{4 m^2-1}{4 m^2+1}, \: \beta=\frac{4 m}{4 m^2+1}
\end{equation*}
These maps are well defined, as $4m^2+1=0$ (mod $d$) has no solution (rsp. $-1$ is no \emph{quadratic residue}) for $d=4k+3$.
\bigskip
 
It was observed by {\sc Balian} and  {\sc Itzykson} (~\cite{BI86}, 1986) that for primes $d= 4k+3$ these matrices provide a \textbf{\emph{unique common orthogonal basis of eigenvectors}} of $\mathbf{F}$.
Up to a phase $\mathbf{F}$ corresponds to $m=\frac{d-1}{2}$. 
See also~\cite{AFN96, AFN97, FL97}. 
\smallskip

They play also a role in the concept 
of \emph{MUB-balanced states}~\cite{ASSW14, ABD14, WS07}.
\medskip

Now we can state the central result of this paper

\newpage

\begin{theorem}                                \label{Theorem}
Let $d=4k+3$ be a prime. Let $\{|\psi_i\rangle: 1 \leq i \leq d\}$ be the unique common orthogonal basis of eigenvectors of the $d$ matrices
\begin{equation*}  
\mathbf{R}_m = \frac{1}{d} \sum_{r=0}^{d-1}\sum_{s=0}^{d-1}  \tau^{m (r^2  + s^2)} \boldsymbol{\mathsf{W}}_{(r,s)}, \quad  0 \leq m \leq d-1
\end{equation*}
The $d^3$ vectors $\{\boldsymbol{\mathsf{W}}_{(k,l)} |\psi_i \rangle: 1 \leq i \leq d \mbox{ , } 0 \leq k,l \leq d-1 \}$ form a $\mathbf{2}$\textbf{-design}.
\end{theorem}

Actually, in the following, we are going to work with the $d$ corresponding projection matrices $\{\mathbf{P}_i:=|\psi_i\rangle \langle \psi_i| : 1 \leq i \leq d\}$ on the eigenvectors. 
Let 
\begin{equation}
\mathbf{P}_i^{(k,l)}=  \boldsymbol{\mathsf{W}}_{(k,l)} \mathbf{P}_i \boldsymbol{\mathsf{W}}_{(k,l)}^{-1} \mbox{ with } 0 \leq i, k,l \leq d-1. 
\end{equation}
We are going to prove
\begin{equation}             \label{E:2DesProp}                      
\frac{1}{d^3}\sum_{i=0}^{d-1} \sum_{k=0}^{d-1} \sum_{l=1}^{d-1} \mathbf{P}_i^{(k,l)} \otimes \mathbf{P}_i^{(k,l)} = \frac{2}{d(d+1)}  \mathbf{\Pi}_{\text{sym}} 
\end{equation}

For this we need some preparation
\begin{itemize}
\item{In the next subsection} we sum up some  features of the Weyl-Heisenberg matrices and of their tensor products, which we use subsequently.
\item{In the subsection afterwards} we expand the projection matrices $\mathbf{P}_i$ on the eigenvectors in terms of the Weyl-Heisenberg matrices. 
We don't calculate them explicitly, but just show some of their properties, based on relations of the Weyl-Heisenberg matrices and the $\mathbf{R}_m$, which allow us 
\item{in the final subsection to prove the theorem}. 
\end{itemize}

\subsection{Some properties of Weyl-Heisenberg matrices}

First we summarize some well-known relations on the Weyl-Heisenberg matrices~\cite[etc. ( $d$ odd)]{Ap05,Ap09,BZ17a,BZ17b}.
\begin{equation} 			  \label{E:WHProp1}
\boldsymbol{\mathsf{W}}_{(k,l)} \boldsymbol{\mathsf{W}}_{(r,s)} = 
\tau^{2(r l-s k)}\boldsymbol{\mathsf{W}}_{(r,s)}\boldsymbol{\mathsf{W}}_{(k,l)} =
\tau^{(r l-s k)}\boldsymbol{\mathsf{W}}_{(k+r,l+s)}
\end{equation}
\begin{equation}                      \label{E:WHProp3}
\tr\left(\boldsymbol{\mathsf{W}}_{(k,l)}\boldsymbol{\mathsf{W}}_{(k',l')}^*\right)=
\begin{cases}
d  & \mbox{if } k= k', l= l'    \: (\mbox{mod } d) \\ 
0   & \mbox{else }.
\end{cases}
\end{equation}
The last equation describes the already mentioned orthogonality of the Weyl-Heisenberg matrices.
Immediate consequence is (see also~\cite{Ya21}) 
\begin{multline}            \label{E:WHProp5}               
\sum_{k=0}^{d-1} \sum_{l=0}^{d-1}  \boldsymbol{\mathsf{W}}_{(k,l)}  \boldsymbol{\mathsf{W}}_{(r,s)} \boldsymbol{\mathsf{W}}_{(k,l)}^{-1} \otimes  \boldsymbol{\mathsf{W}}_{(k,l)}  \boldsymbol{\mathsf{W}}_{(r',s')} \boldsymbol{\mathsf{W}}_{(k,l)}^{-1} \\
=  
\begin{cases}
d^2\left(\boldsymbol{\mathsf{W}}_{(r,s)} \otimes  \boldsymbol{\mathsf{W}}_{(-r,-s)}\right)  	& \mbox{if } r'=-r, s'=-s   \: (\mbox{mod } d)\\
0 																	& \mbox{else}.
\end{cases}
\end{multline}

We furthermore notice the following formula
\begin{equation*}  
\mathsf{SWAP}=
\sum_{q_1=0}^{d-1}\sum_{q_2=0}^{d-1} ( |\mathbf{e}_{q_2}\rangle\langle\mathbf{e}_{q_1}| \otimes |\mathbf{e}_{q_1}\rangle\langle\mathbf{e}_{q_2}| ) =
\frac{1}{d}\sum_{r=0}^{d-1}\sum_{s=0}^{d-1}  \boldsymbol{\mathsf{W}}_{(r,s)} \otimes \boldsymbol{\mathsf{W}}_{(-r,-s)}    
\end{equation*} 
This is a special case of $\mathsf{SWAP}= \frac{1}{d}\sum_{r=0}^{d-1}\sum_{s=0}^{d-1}  \mathbf{g}_{(r,s)} \otimes \mathbf{g}_{(r,s)}^{*}$  for any orthogonal matrix base $\{ \mathbf{g}_{(r,s)}\}$, $0 \leq r, s \leq d-1$, were the standard norm $\|\mathbf{g}_{(r,s)}\|^2_2=d$. See the paper by {\sc Siewert}~\cite{Si22}, were this relation is extensively exploited.  An immediate consequence is 
(see also~\cite{OY19}:
\begin{equation}  				 \label{E:WHProp-Sym}   
\mathbf{\Pi}_{\text{sym}} = \frac{1}{2}\left(\mathbf{I} \otimes \mathbf{I} +\mathsf{SWAP} \right) =
 \frac{1}{2}\left(\mathbf{I}  \otimes \mathbf{I} +\frac{1}{d}\sum_{r=0}^{d-1}\sum_{s=0}^{d-1}  \boldsymbol{\mathsf{W}}_{(r,s)} \otimes \boldsymbol{\mathsf{W}}_{(-r,-s)} \right)
\end{equation} 

\medskip

\subsection{Some properties of $\mathbf{R}_m$ and the projection matrices $\mathbf{P}_i$}

\begin{lemma}                                 \label{L:RmProp}
Let $d=4k+3$ be prime, and $\mathbf{R}_m$ as in definition~\eqref{E:Rm}. For $0 \leq m, m' \leq d-1$
\begin{equation}                      \label{E:RmProp}
\tr\left(\mathbf{R}_m \right)= 1,  \qquad  \qquad
\tr\left(\mathbf{R}_m \mathbf{R}_{m'}^*\right)=
\begin{cases}
-1  & \mbox{if } m \neq m'  \\ 
d   & \mbox{if } m=m' .
\end{cases}
\end{equation}
\end{lemma} 
\begin{proof}[Proof]
$\tr\left(\mathbf{R}_{m}\right)= 1$ follows directly, and further using~\eqref{E:WHProp3}, after short calculation
\begin{equation*}
\tr\left(\mathbf{R}_{m} \mathbf{R}_{m'}^*\right) =
\frac{1}{d}\left( \sum_{r=0}^{d-1} \tau^{(m-m')r^2}\right)^2
=
\begin{cases}
(\pm i)^2=-1  & \mbox{if } m \neq m'   \\ 
d   & \mbox{if } m=m' .
\end{cases}
\end{equation*}
Here we used, that for prime numbers $d=4k+3$ the \emph{quadratic Gauss sum} are $\sum_{r=0}^{d-1} \tau^{a r^2}= \left( \frac{a}{d} \right) i \sqrt{d}$ with the \emph{Legendre symbol}
 $\left( \frac{a}{d} \right)=\pm 1$, for $a \neq 0$. 
\end{proof}
\bigskip
\smallskip

Next we define auxiliary matrices $\mathbf{X}_{m}$ as \emph{orthonormalization} of the $\mathbf{R}_{m}$.
\begin{lemma}                                 \label{L:XmProp}
Let $d=4k+3$ be prime, and 
\begin{equation}                      \label{E:XmProp1}
\mathbf{X}_{m} := \frac{1}{\sqrt{d+1}} \mathbf{R}_{m} +  \kappa  \mathbf{I} , \quad  \kappa=\frac{\sqrt{d+1}-1}{d\sqrt{d+1}}, \quad 0 \leq m \leq d-1
\end{equation}
then for $0 \leq m, m' \leq d-1$
\begin{equation}                      \label{E:XmProp2}
\tr\left(\mathbf{X}_{m}\right)= 1,  \qquad  \qquad
tr\left(\mathbf{X}_{m} \mathbf{X}_{m'}^*\right)=
\begin{cases}
0  & \mbox{if } m \neq m'  \\ 
1  & \mbox{if } m=m' .
\end{cases}
\end{equation}
\end{lemma} 
\begin{proof}[Proof]
The proof is straight forward using the definition and equations~\eqref{E:RmProp}.
\end{proof}

Next we derive an \emph{Ansatz} for the projection matrices on the eigenvectors of  $\mathbf{R}_{m}$ via the auxiliary matrices $\mathbf{X}_{m}$. This helps us
to prove some properties.

\begin{lemma}                                 \label{L:PProp}
Let $d=4k+3$ be prime, and let
\begin{equation}			\label{E:PDef1}
\mathbf{P}_{i} :=   
 \sum_{m=0}^{d-1} \lambda_{i m} \mathbf{X}_{m}  \quad  0 \leq i,m \leq d-1
\end{equation}
be the projection matrices on the common eigenvectors of all $\mathbf{R}_{m}$ in arbitrary order. Then for all $0 \leq i \leq d-1$
\begin{equation}			\label{E:PDef2}
\mathbf{P}_{i} 
= \frac{1}{d\sqrt{d+1}} \sum_{r=0}^{d-1}\sum_{s=0}^{d-1} p_i^{(r,s)} \boldsymbol{\mathsf{W}}_{(r,s)} + \kappa \mathbf{I} 
\end{equation}
with
\begin{equation}			\label{E:PDef3}
\kappa=\frac{\sqrt{d+1}-1}{d\sqrt{d+1}}, \qquad
p_i^{(r,s)}=\sum_{m=0}^{d-1} \lambda_{i m}  \tau^{m (r^2  + s^2)}  \quad  0 \leq r,s,i \leq d-1
\end{equation} 
and we have
\begin{subequations} \label{E:PProp}
\begin{gather}
\qquad p_i^{(0,0)} =1 \qquad \qquad 0 \leq i \leq d-1   								\label{E:PProp1} \\
\: \; \; \quad \quad (p_i^{(r,s)})^* = p_i^{(-r,-s)} \quad \; 0 \leq r,s,i \leq d-1					\label{E:PProp2} \\
\|(p_i^{(r,s)})_{0 \leq i \leq d-1}\|^2_2  = d \qquad \qquad  0 \leq r,s \leq d-1 \qquad	\:		\label{E:PProp3} 
\end{gather}
\end{subequations}
\end{lemma} 
\medskip

\begin{proof}[Proof]
The matrix $\mathbf{\Lambda}= (\lambda_{i m})_{0 \leq i \leq d-1; \; 0 \leq m \leq d-1}$ transfers between the $d$ orthonormal matrices $\mathbf{X}_{m}$ and the also orthonormal $d$ matrices $\mathbf{P}_{i}$. Therefore it must be unitary, which we use below. 
\smallskip

We insert $\mathbf{X}_{m}$ rsp. $\mathbf{R}_{m}$  in definition~\eqref{E:PDef1} and get
\begin{equation*}
\mathbf{P}_{i} 
= \frac{1}{d\sqrt{d+1}} \left( \sum_{r=0}^{d-1} \sum_{s=0}^{d-1}\sum_{m=0}^{d-1} \lambda_{i m}  \tau^{m (r^2  + s^2)} \boldsymbol{\mathsf{W}}_{(r,s)} \right) +  \sum_{m=0}^{d-1} \lambda_{i m} \kappa \mathbf{I} 
\end{equation*}
By applying the trace on~\eqref{E:PDef1}, and using $\tr(\mathbf{P}_{i})=1$ and equation~\eqref{E:XmProp2} we get
$\sum_{m=0}^{d-1} \lambda_{i m} =1$, for all $ 0 \leq i \leq d-1$.
This proves~\eqref{E:PDef2}
with the definition of $p_i^{(r,s)}$ as in~\eqref{E:PDef3}. Now for the properties of these coefficients.
\begin{itemize}
\item{$p_i^{(0,0)} =1$} follows again from $\sum_{m=0}^{d-1} \lambda_{i m} =1$, for all $ 0 \leq i \leq d-1$

\item{$ (p_i^{(r,s)})^* = p_i^{(-r,-s)}$} follows as $\mathbf{P}_{i}^*=\mathbf{P}_{i}$ and $\boldsymbol{\mathsf{W}}_{(r,s)}^*=\boldsymbol{\mathsf{W}}_{(r-,-s)}$.

\item{$\|(p_i^{(r,s)})_{0 \leq i \leq d-1}\|^2_2  = d$} is valid, as the vector $(p_i^{(r,s)})_{0 \leq i \leq d-1}$ is for any $0 \leq r,s \leq d-1$ the unitary transform of the vector  $(\tau^{m (r^2  + s^2)})_{0 \leq m \leq d-1}$ by $\mathbf{\Lambda}$, 
and this vector has obviously squared norm $\|(\tau^{m (r^2  + s^2)})_{0 \leq m \leq d-1}\|^2_2 = d$.

\end{itemize}
\end{proof}
\emph{Remark:} When $\mathbf{P}_{x}$ is an \emph{fiducial} projector for a Weyl-Heisenberg covariant SIC, then $|p_x^{(r,s)}|^2  = 1$. 
See also~\cite{ABFG19, OY19, Ya21} for the SIC-related background.

\newpage
\subsection{Proof of Theorem 1}
\begin{proof}[We want to prove equation~\eqref{E:2DesProp}] 
Let
\begin{equation*} 
x=                        
\frac{1}{d^3}  \sum_{i=0}^{d-1} \sum_{k=0}^{d-1} \sum_{l=0}^{d-1} \mathbf{P}_i^{(k,l)} \otimes \mathbf{P}_i^{(k,l)}  =
\frac{1}{d^3} \sum_{i=0}^{d-1} \sum_{k=0}^{d-1} \sum_{l=0}^{d-1}  \boldsymbol{\mathsf{W}}_{(k,l)} \mathbf{P}_i \boldsymbol{\mathsf{W}}_{(k,l)}^{-1} \otimes  \boldsymbol{\mathsf{W}}_{(k,l)} \mathbf{P}_i \boldsymbol{\mathsf{W}}_{(k,l)}^{-1}
\end{equation*}
When we insert $\mathbf{P}_i$, as described in Lemma~\ref{L:PProp}, the sums reduce significantly, due to the equation~\eqref{E:WHProp5} about sums of tensor-products of Weyl-Heisenberg matrices. We get
\begin{multline*}                         
x                        
=\frac{1}{d^3 (d+1)} \sum_{i=0}^{d-1} 
\sum_{r=0}^{d-1}\sum_{s=0}^{d-1}  
p_i^{(r,s)} p_i^{(-r,-s)} \boldsymbol{\mathsf{W}}_{(r,s)}  \otimes  
\boldsymbol{\mathsf{W}}_{(-r,-s)}
\\
+ \sum_{i=0}^{d-1} p_i^{(0,0)} \frac{2(\sqrt{d+1}-1)}{d^3 (d+1)}  \mathbf{I}  \otimes \mathbf{I} 
+ \sum_{i=0}^{d-1} \frac{(\sqrt{d+1}-1)^2}{d^3 (d+1)} \mathbf{I} \otimes  \mathbf{I}
\end{multline*}
We use~\eqref{E:PProp1} $p_i^{(0,0)} =1$ and ~\eqref{E:PProp2} $p_i^{(r,s)}p_i^{(-r,-s)} = |p_i^{(r,s)}|^2$ and get
\begin{equation*}                         
x                        
=\frac{1}{d^3 (d+1)} 
\sum_{i=0}^{d-1} \sum_{r=0}^{d-1}\sum_{s=0}^{d-1}  
|p_i^{(r,s)}|^2 \boldsymbol{\mathsf{W}}_{(r,s)}  \otimes  
\boldsymbol{\mathsf{W}}_{(-r,-s)}
+\frac{1}{d (d+1)} \mathbf{I} \otimes  \mathbf{I}
\end{equation*}
And finally we use the key equation~\eqref{E:PProp3} $\|(p_i^{(r,s)})_{0 \leq i \leq d-1}\|^2_2=|p_0^{(r,s)}|^2 + \ldots + |p_{d-1}^{(r,s)}|^2=d$ to get
\begin{equation*}                         
x                         
=\frac{1}{d (d+1)} \left(\frac{1}{d}
\sum_{r=0}^{d-1}\sum_{s=0}^{d-1}  \boldsymbol{\mathsf{W}}_{(r,s)}  \otimes  
\boldsymbol{\mathsf{W}}_{(-r,-s)}
+ \mathbf{I} \otimes  \mathbf{I} \right) 
= \frac{2}{d(d+1)}  \mathbf{\Pi}_{\text{sym}} 
\end{equation*}
according the formula for $\mathbf{\Pi}_{\text{sym}}$ in~\eqref{E:WHProp-Sym}.
\end{proof}

\section{State space}

{\sc Słomczyński} and  {\sc Szymusiak}~\cite{SS19} noticed, that $\{\mathbf{P}_i:=|\psi_i\rangle \langle \psi_i| : 1 \leq i \leq N\}$ is a $2$-design, iff
for any matrix $\boldsymbol{\rho}$ with $\tr(\boldsymbol{\rho})=1$
\begin{equation*} 
\boldsymbol{\rho}= (d+1)              
\sum_{i=1}^{N} \frac{d}{N}\tr(\boldsymbol{\rho}\mathbf{P}_i)\mathbf{P}_i-\mathbf{I}
\end{equation*}
This equation can e.g. be proven by using equation~\eqref{E:2-Des-Prop} for $2$-designs, with 
$\mathbf{\Pi}_{\text{sym}} = \frac{1}{2}\left(\mathbf{I} \otimes \mathbf{I} +\mathsf{SWAP} \right)$ 
\begin{equation*}                        
\sum_{i=1}^N \mathbf{P}_i \otimes \mathbf{P}_i = \frac{N}{d(d+1)} \left(\mathbf{I} \otimes \mathbf{I} +\mathsf{SWAP} \right). 
\end{equation*}
Multiplying it with $\mathbf{I} \otimes \boldsymbol{\rho}$, applying the partial trace, and using ({\sc Siewert}~\cite{Si22})
\begin{equation*}                
\tr_{[2]}\left((\mathbf{I} \otimes\boldsymbol{\rho}) \cdot \mathsf{SWAP} \right)=\boldsymbol{\rho}
\end{equation*} 
we get $           
\sum_{i=1}\tr(\boldsymbol{\rho}\mathbf{P}_i)\mathbf{P}_i=\frac{N}{d(d+1)}(\tr(\boldsymbol{\rho})\mathbf{I}+\boldsymbol{\rho})
$.

\newpage

As a consequence we state
\begin{corollary}                                \label{Corollary}
Let $d=4k+3$ be a prime. Let  $\{\mathbf{P}_i:=|\psi_i\rangle \langle \psi_i| : 1 \leq i \leq d\}$ be the projection matrices on the unique common orthogonal basis of eigenvectors of the $d$ matrices
$\mathbf{R}_m$ as in Theorem~\ref{Theorem}
and 
$
\mathbf{P}_i^{(k,l)}=  \boldsymbol{\mathsf{W}}_{(k,l)} \mathbf{P}_i \boldsymbol{\mathsf{W}}_{(k,l)}^{-1} \mbox{ with } 0 \leq i, k,l \leq d-1. 
$
Then for any matrix $\boldsymbol{\rho}$ with $\tr(\boldsymbol{\rho})=1$ (e.g. \emph{density matrices})
\begin{equation} 
\boldsymbol{\rho}= (d+1)              
\sum_{i=0}^{d-1}\sum_{j=0}^{d-1}\sum_{k=0}^{d-1} \rho_i^{(k,l)} \mathbf{P}_i^{(k,l)}-\mathbf{I} \qquad \mbox{ with }  \rho_i^{(k,l)}=\frac{1}{d^2}\tr(\boldsymbol{\rho}\mathbf{P}_i^{(k,l)})
\end{equation}
\end{corollary}
\bigskip

This result can be seen as counterpart of the equation for the $d^2$ projection matrices of a  \emph{SIC} $2$-designs 
$\boldsymbol{\rho}= (d+1)              
\sum_{i=1}^{d^2} \rho_i \mathbf{P}_i-\mathbf{I}$, with 
$\rho_i=\frac{1}{d}\tr(\boldsymbol{\rho}\mathbf{P}_i)$ 
which (as \emph{"measurement in the sky"}) attracted a great deal of attention in the context of the \emph{QBism} approach to quantum theory by {\sc Fuchs} et al.~\cite{FS09, FS19}.

\bigskip

\section{Concluding remarks}

One can try to generalize the construction given here to further dimensions $d$. 
The rank $1$ projection matrices of SICs, as well as complete sets of MUBs are projective $2$-designs, which, like the construction here, are covariant under the Weyl-Heisenberg groups. 
Are there more similar $2$-designs?

Finally it would be interesting, if and how the $2$-design property could possibly be related to the quantum harmonic oscillator.
\medskip

An outline of this paper was presented at the 6th Workshop on Algebraic Designs, Hadamard Matrices $\&$ Quanta, at Jagiellonian University Krakow, June 27 – July 2, 2022.

The author is grateful to Marcus Appleby, Markus Grassl and Danylo Yakymenko for valuable comments.

\end{document}